\documentclass[aps,pra,showpacs,twocolumn]{revtex4}

\usepackage{amsmath}
\usepackage{amssymb}
\usepackage{epsfig}
\usepackage{color}
\usepackage{bbm}

\graphicspath{{.}{./EPS/}}

\newcommand* {\bra}[1]{\ensuremath{\langle {#1} |}}
\newcommand* {\ket}[1]{\ensuremath{| {#1} \rangle}}
\newcommand* {\ee}{\ensuremath{\mathrm{e}}}
\newcommand* {\cU}{\ensuremath{\mathcal{U}}}
\newcommand* {\cH}{\ensuremath{\mathcal{H}}}
\newcommand* {\cB}{\ensuremath{\mathcal{B}}}

\newcommand* {\id}{\ensuremath{\mathbbm{1}}}

\newtheorem{theorem}{Theorem}[section]

\newtheorem{proposition}[theorem]{Proposition}
\newtheorem{corollary}[theorem]{Corollary}

\newenvironment{proof}[1][Proof]{\begin{trivlist}
\item[\hskip \labelsep {\bfseries #1}]}{\end{trivlist}}

\newenvironment{remark}[1][Remark]{\begin{trivlist}
\item[\hskip \labelsep {\bfseries #1}]}{\end{trivlist}}

\newcommand{\qed}{\nobreak \ifvmode \relax \else
      \ifdim\lastskip<1.5em \hskip-\lastskip
      \hskip1.5em plus0em minus0.5em \fi \nobreak
      \vrule height0.75em width0.5em depth0.25em\fi}

\begin{document}
 
\title{Effects of stochastic noise on dynamical decoupling procedures}
\author{J. Z. Bern\'ad}
\email{Zsolt.Bernad@physik.tu-darmstadt.de}
\affiliation{Institut f\"{u}r Angewandte Physik, Technische Universit\"{a}t Darmstadt, D-64289, Germany}
\author{H. Frydrych}
\affiliation{Institut f\"{u}r Angewandte Physik, Technische Universit\"{a}t Darmstadt, D-64289, Germany}

\date{\today}

\begin{abstract}
Dynamical decoupling is an important tool to counter decoherence and dissipation effects in quantum systems originating from 
environmental interactions. It has been used successfully in many experiments; however, there is still a gap between fidelity improvements 
achieved in practice compared to theoretical predictions. We propose a model for imperfect dynamical decoupling based on a stochastic Ito 
differential equation which could explain the observed gap. We discuss the impact of our model on the time evolution of various quantum systems in 
finite- and infinite-dimensional Hilbert spaces. Analytical results are given for the limit of continuous control, whereas we present numerical 
simulations and upper bounds for the case of finite control.
\end{abstract}


\maketitle

\section{Introduction}
Dynamical decoupling is a well-established technique to suppress the effects of an unwanted Hamiltonian interaction on a quantum system, 
whereby the Hamiltonian can describe inner-system interactions as well as interactions with an environment 
\cite{Viola1,Viola2,Zanardi,Vitali}. 
The method works by asserting active control over the system: at certain times unitary operations are performed on the system and rotate it
 in its quantum state space such that over time the effects of the Hamiltonian cancel. 
Dynamical decoupling has been formally introduced in a general way by Viola {\it et al.} \cite{Viola1,Viola2}, but the basic principle was known 
and used before, particularly in the nuclear magnetic resonance community \cite{hahn,Carr,Meiboom,Haeberlen}, and 
the spin-echo effect \cite{hahn} is probably the best-known example of a simple decoupling procedure.

For an arbitrary quantum system the sequence of unitary operators needed to protect it can become fairly complex, and in general dynamical 
decoupling works only approximately in that it decouples the acting Hamiltonian only up to the first order in the Magnus expansion 
\cite{Magnus} of the average time evolution Hamiltonian. While the pioneering studies described dynamical decoupling with equidistant 
deterministic pulse sequences, a variety of approaches to the original idea were developed in order to optimize these sequences further 
\cite{Uhrig,Khodjatesh,Witzel,Lee,Gordon,Souza}. The concept of randomly chosen pulses was introduced as well \cite{Kern1,Viola3}, 
and it was shown that under certain conditions random dynamical decoupling could outperform the deterministic approach \cite{Santos}. 
A hybrid scheme combining the advantages of both approaches allows for additional improvements \cite{Kern2}. 
In the case of the 
protected system being a simple two-state quantum system, i.e., a qubit, dynamical decoupling techniques were applied with success in many 
experimental works \cite{Biercuk1,Biercuk2,Du,Barthel,Bluhm,deLange,Bylander,Damodarakurup,Lucamarini}.

Dynamical decoupling schemes are particularly promising in the field of quantum computing where decoherence poses a serious hurdle in experiments.
Decoupling techniques are able to drastically improve the coherence time of qubit states and of single- or two-qubit operations over the 
unprotected case. In theory, the achievable fidelity is limited only by the strength of the uncorrected higher orders of the average Hamiltonian and the time between decoupling pulses. In practice, however, we have to expect an additional fidelity decay from experimental realities, such as additional interactions on the 
qubits outside of the studied Hamiltonian interaction or imperfections in the measurement. Furthermore, we also have to consider imperfect decoupling controls. In theory it is often assumed that the decoupling pulses are perfect and happen instantaneously (infinitesimal width), which is an assumption that cannot hold in practice. 
In \cite{Khodjatesh2}, the effects of finite-width pulses were studied. The article also introduced a jitter pulse with uniformly distributed random parameters and examined its effects on a concatenated dynamical decoupling sequence.

In this paper we propose a different model for imperfect dynamical decoupling control described by a linear quantum stochastic differential equation with a standard Wiener process. We conduct a rigorous analysis of its effects on dynamical decoupling performance. In particular, we derive the generator of the time evolution and study its form in the limit of continuous control, which allows us to predict the robustness of different decoupling sequences against decoherence in our error model. Numerical simulations for finite numbers of pulses and upper bounds on the rate of convergence are also provided. 
While the model is phenomenological in nature, it is our hope that the results can help gain a deeper understanding of real-world decoupling procedures.
Our work is particularly inspired by a recent experimental publication \cite{Piltz} which demonstrated, for different pulse sequences, that
individual pulse imperfections can accumulate, but also compensate each other. Our model offers one possible explanation of these findings.

The paper is organized as follows. In Sec. \ref{sec:decoupling} we briefly review the dynamical decoupling method and introduce our stochastic noise model. We derive the resulting time evolution with the help of Ito calculus by using an approach similar to, but more general than, the idea developed in \cite{Facchi}. In Sec. \ref{sec:examples} we explore the effects of our noise model on a few exemplary systems in both finite- and infinite-dimensional Hilbert spaces. We present upper bounds for the convergence of our model depending on the number of pulses employed. 
Detailed derivations and proofs which support the main text are collected in two Appendixes.

\section{Decoupling schemes with stochastic noise}  \label{sec:decoupling}

In the general formulation of dynamical decoupling we consider a closed quantum system under the influence of a constant 
Hamiltonian. The goal is to control or protect a subpart of the full system, with the rest being treated as an environment. This 
is usually done by eliminating the interaction between the considered subpart and the rest of the system, while the free evolution of 
the subpart remains intact. In order to achieve this purpose we assume that, after fixed times  $n \tau$, a unitary pulse $u_n$ can be 
applied to the subsystem instantaneously (\emph{bang-bang} control \cite{Viola1}). 
With the right choice of pulses it is possible to drastically improve the lifetimes of quantum states in the subspace, and in the limit of 
continuous control $\tau \to 0$ we find perfect coherence preservation.

In the following we will model imperfect pulses $u_n$ to account for the fact that perfect instantaneous unitary pulses are not available in actual 
experiments. 
In our model we assume that the apparatus implementing the decoupling pulses $u_n$ causes a stochastic error over time to the controlled system. 
Mathematically we introduce this error by means of an Ito stochastic differential equation for the applied unitary pulses:
\begin{equation}
\label{basic}
du_n(t) = \left( -\frac{\gamma}{2} B^2 dt - i \sqrt{\gamma} B dW_t \right) u_n(t), \quad u_n(0) = u_n, 
\end{equation}
where $u_n(0)$ is the original ideal pulse, the parameter $\gamma$ stands for the strength of the disturbance, and 
$W_t$ is a classical Wiener process. The operator $B$ stands for a physical quantity, i.e., a self-adjoint operator, 
which describes the nature of the disturbance. $B^2$ is always a positive operator \cite{Rudin}. 
$W_t$ is defined as a Gaussian random variable with expectation value 0 and variance $t$:
\begin{eqnarray}
 &&\mathbb{E}[W_t]=0, \nonumber \\
 &&\mathbb{E}[W^2_t]-\mathbb{E}[W_t]^2=t.
\end{eqnarray}
In order to work with this equation we employ the quantum Ito rules \cite{Parthasarathy} in the sense that the differential equation deals with operators; however, the Wiener process used is still a classical and not an operator-valued process.

We need to ensure that the differential equation \eqref{basic} always results in a unitary operator for $t>0$. 
Using the properties of the 
Wiener process
\begin{equation}
d^2W_t=dt,\quad d^nW_t=0, \; n>2, 
\end{equation}
a straightforward calculation yields
\begin{eqnarray}
 &&d\left(u^\dagger_n(t)\,u_n(t)\right)=du^\dagger_n(t)\,u_n(t)+u^\dagger_n(t)\,du_n(t)+\nonumber \\
&&+du^\dagger_n(t)\,du_n(t)=0=d\left(u_n(t)\,u^\dagger_n(t)\right),
\end{eqnarray}
which means that the solution $u_n(t)$ is unitary for all $t\geqslant0$.

Our model ensures that the distance between the ideal pulse $u_n$ and the average of $u_n(t)$ (over a large sample size) increases with $t$. 
From the property
$\mathbb{E}[dW_t]=0$ it follows for the averaged pulse
\begin{equation}
\label{avbasic}
 \frac{d\mathbb{E}\left[u_n(t)\right]}{dt}=-\frac{\gamma}{2}B^2 \mathbb{E}\left[u_n(t)\right].
\end{equation}
The equation clearly describes a decaying effect, due to the fact that $B^2$ is positive. As an example, consider $B^2$ to be the identity
operator and the distance to be quantified by the operator norm, 
\begin{equation}
 ||A||_{op}=\sup\limits_{||x||=1}||Ax||,\,\,\,\,A \in {\cal B}(\mathcal{H}),
\end{equation}
where ${\cal B}(\mathcal{H})$ is the set of all bounded linear operators on the Hilbert space $\mathcal{H}$, and the vector norm $||\cdot||$ is 
generated by the inner product. In this case we have
\begin{equation}
 ||u_n-\mathbb{E}[u_n(t)]\,||_{op}=||\left(1-\ee^{-\frac{\gamma t}{2}}\right)u_n||_{op}=1-\ee^{-\frac{\gamma t}{2}}.
\end{equation}

We now derive the time evolution of the full system generated by a Hamiltonian $H$ and $N$ decoupling pulses, applied evenly over the whole
interaction time $t$ with a distance between pulses of $\tau=t/N$ and their imperfection governed by Eq. \eqref{basic}. After time 
$t$ the time evolution is governed by the unitary operator
\begin{equation}
\label{decouplingproc}
U_N(t) = \prod_{k=0}^{N-1} u_{N-k}(\frac{\tau}{N}) \ee^{-iH\tau} .
\end{equation}
The time parameter of $\frac{\tau}{N}$ for the pulses is a mathematical consequence of our intended noise model and should not be understood as the 
physical implementation time of the pulses. We want to study a global 
noise induced by the pulse-generating device which is governed by $W_t$. Mathematically we express the error by partitioning it over 
the $N$ pulses applied to the system, and due to the particular rescaling properties of the Wiener process we have $W_t = 1/\sqrt{a} W_{at}$ 
for $a>0$ and therefore
\begin{equation}
\ee^{-i W_t} = \ee^{-i N W_{t/N^2}} = \left(\ee^{-iW_{\tau/N}} \right)^N .
\end{equation}
Our model is phenomenological in nature. Since all pulses share the same random variable $W_t$, they will introduce the same error to the system during a single run of the experiment. However, the error will be different for subsequent repetitions of the experiment. As such, this model captures imperfections due to changing conditions between experimental runs, which cannot typically be avoided perfectly or might even happen deliberately, e.g., through necessary recalibrations of the experimental apparatus.

Our aim is to calculate the derivative $dU_N(t)$ in order to determine the time evolution of the system with $N$ stochastic pulses applied. 
We apply the Ito formula for an $N$-term product, which states
\begin{eqnarray}
 &&d(A_1 \cdots A_N)=d(A_1) A_2 \cdots A_N + \cdots + A_1 \cdots d(A_N) \nonumber \\
&&+\sum_{\text{all possible pairings}} A_1 \cdots d(A_i) \cdots d(A_j) \cdots A_N .
\end{eqnarray}
We introduce the operators
\begin{equation}
g_n(\tau) = \prod_{k=0}^{n-1} u_{n-k}(\frac{\tau}{N}) \ee^{-iH\tau} , \quad g_n := g_n(0), \quad g_0(\tau) := \mathbbm{1},
\end{equation}
and find
\begin{eqnarray}
dU_N(t) = \Big(&-&i H_N(t) dt - i \sqrt{\gamma} B_N(t) dW_t -\nonumber \\
&-&\frac{\gamma}{2} C_N(t) dt \Big)U_N(t), 
\end{eqnarray}
where
\begin{eqnarray}
H_N(t) &=& \frac{1}{N} \sum_{k=1}^N g_{k-1}(\tau) u_{N-k+1}(\frac{\tau}{N}) H u_{N-k+1}^\dag(\frac{\tau}{N}) g_{k-1}^\dag(\tau), \nonumber \\
\label{HN}\\
B_N(t) &=& \frac{1}{N} \sum_{k=0}^{N-1} g_k(\tau) B g_k^\dag(\tau), \label{BN}\\
C_N(t) &=& \frac{1}{N^2} \sum_{k=0}^{N-1} g_k(\tau) B^2 g_k^\dag(\tau)+ \label{CN} \\
  &+& \frac{2}{N^2} \sum_{j=0}^{N-2} \sum_{k=j+1}^{N-1} g_j(\tau) B g_j^\dag(\tau) g_k(\tau) B g_k^\dag(\tau), \nonumber
\end{eqnarray}
$H_N(t)$ is the Hamiltonian operator generating the error-free time evolution. 
The operators $B_N(t)$ and $C_N(t)$ express the error due to the stochastic noise and are related as
\begin{equation}
\frac{C_N(t) + C_N^\dag(t)}{2} = B^2_N(t).
\end{equation}

We are particularly interested in the limit of $N \to \infty$, the limit of continuous control. The time evolution in this limit, 
\begin{equation}
d\cU(t) = -i\cH \cU(t) dt - i\sqrt{\gamma} \cB \cU(t) dW_t - \frac{\gamma}{2} \cB^2 \cU(t) dt,
\label{TevU}
\end{equation}
with
\begin{equation}
\cU(t=0)=\lim_{N \to \infty} \prod_{k=0}^{N-1} u_{N-k},  
\end{equation}
defines the time evolution of any observable $O(t)=\cU^\dag(t)O \cU(t)$ and of any density matrix $\rho(t)=\cU(t) \rho \cU^\dag(t)$.
The initial condition $\cU(t=0)$ is a unitary operator due to the fact that unitary operators are closed under multiplication. Most of the 
experiments are designed such that $\cU(t=0)=\id$. We introduced the following notations:
\begin{eqnarray}
\cU(t) &=& \lim_{N\to\infty} U_N(t),\,\,\,  \cH = \lim_{N\to\infty} H_N(t),  \\
\cB &=& \lim_{N\to\infty} B_N(t), \,\,\,\, \cB^2 = \lim_{N\to\infty} C_N(t). 
\end{eqnarray}
Using the Ito rules, a straightforward 
calculation for the dynamical evolution of the density matrix yields
\begin{equation}
\label{maineq}
d\rho(t)=\mathcal{L}\rho(t)dt-i\sqrt{\gamma} \left[\cB,\rho(t)\right] dW_t,
\end{equation}
where the Kossakowski-Lindblad generator \cite{Gorini,Lindblad} is given by
\begin{equation}
 \mathcal{L}\rho=-i\left[\cH, \rho\right]+\gamma \cB \rho \cB-\frac{\gamma}{2}
\left(\cB^2 \rho + \rho \cB^2\right).
\end{equation}

The most interesting scenario for both theoretical and experimental studies is the case of finite decoupling 
sequences repeatedly applied in cycles. We consider the sequence length to be $M>0$ and $g_M=U$. The number of cycles is taken to be 
$L$ and $N=L M$. Now, in the $N \to \infty$ limit, the generator of Eq. \eqref{TevU} can be expressed as the limit of
a sequence of Ces\`{a}ro means. This can be easily seen if we take $\tau \to 0$ in Eqs. \eqref{HN},\eqref{BN}, and \eqref{CN}.
Therefore, we can associate the limiting procedure for the generator 
to von Neumann's mean ergodic theorem which describes the limit as a projection operator (see Appendix \ref{appendix1}).

If $H$ and $B$ are bounded operators, the formulas for $\cH$ and $\cB$ in the limit $N \to \infty$ result in
\begin{eqnarray}
\cH&=&\mathcal{P}\left(\frac{1}{M} \sum_{j=0}^{M-1} g_j H g_j^\dag\right), \nonumber \\ 
\cB&=&\mathcal{P}\left(\frac{1}{M} \sum_{j=0}^{M-1} g_j B g_j^\dag\right), \label{generallimit}
\end{eqnarray}
where $\mathcal{P}$ is the projector onto the set $\{X \in {\cal B}(\mathcal{H})|UX=XU\}$.

In the special case $g_M = \id$ and generally $g_n = g_{n \bmod M}$, the projection operator $\mathcal{P}$ is simply the identity map on 
${\cal B}(\mathcal{H})$, and we get
\begin{align}
\cH &= \frac{1}{M} \sum_{j=1}^M g_j H g_j^\dag \label{grouplimit}.
\end{align}
This is the common decoupling condition due to Viola {\it et al.} \cite{Viola2}. 
The limit of the error operator also simplifies to
\begin{equation} \label{B}
\cB = \frac{1}{M} \sum_{j=1}^M g_j B g_j^\dag. 
\end{equation}
The structure of $\cH$ and $\cB$ is formally equivalent and suggests the possibility that a carefully designed decoupling scheme may 
eliminate both operators.

Another special case is a single decoupling operator $U$.
Now, we have the following simplified equations:
\begin{equation}
 \cH=\mathcal{P}(H), \quad \cB=\mathcal{P}(B).
\end{equation}
This case is widely used with the condition $U^2 = \id$.

If the operators are unbounded, their domain must be taken into consideration. Equation \eqref{maineq} can be defined only on the overlap 
$A=\mathrm{Dom}(H) \cap \mathrm{Dom}(B)$, and now 
$\cH, \cB \in \{X|\mathrm{Dom}(X)\subseteq A ,UX=XU\}$.

\section{Examples}  \label{sec:examples}
In this section we study the effects of our imperfect decoupling pulses on various quantum systems. First we look at two 
qubits coupled by a Hamiltonian interaction. Despite this system's simplicity we are able to derive a number of properties of our noise model which 
carry over to more complex systems. This is demonstrated in our study of an electron spin coupling to a nuclear spin bath. Finally, we briefly 
look at two coupled harmonic oscillators as an example of a quantum system in an infinite-dimensional Hilbert space. 

\subsection{Two coupled qubits}
We consider two coupled qubits with energy eigenstates
$\ket{0}_i$ and $\ket{1}_i$ ($i \in \{1,2\}$) under the influence of the Hamiltonian
\begin{equation} \label{toymodel}
H=\frac{\hbar \omega}{2} \sigma^{(1)}_z  \otimes \id^{(2)} + \frac{\hbar \omega}{2} \id^{(1)} \otimes 
\sigma^{(2)}_z 
+ g \sigma^{(1)}_x \otimes 
\sigma^{(2)}_x,
\end{equation}
 with $\hbar \omega$ the energy difference between the levels of the qubits and $g$ the coupling constant between the two qubits. We also introduced the notations 
$\sigma^{(i)}_x=\ket{1}_{ii} \bra{0}+\ket{0}_{ii} \bra{1}$ and 
$\sigma^{(i)}_z=\ket{1}_{ii} \bra{1} - \ket{0}_{ii} \bra{0}$ ($i \in \{1,2\}$). 
By means of dynamical decoupling we want to protect the evolution of qubit $1$ against the effects of the interaction with 
qubit $2$, which induce transitions between the energy eigenstates of the two qubits.
Although this model is very primitive from a physical viewpoint, it is very useful to showcase quite a few properties of our noisy decoupling. These properties can then directly be applied to more sophisticated systems.

For our decoupling scheme we choose a single unitary pulse of the form
\begin{eqnarray} \label{tq_pulse}
 U&=&\sigma^{(1)}_z \otimes \id^{(2)}_2,
\end{eqnarray}
resulting in two scheme operators $g_1 = U$ and $g_2 = \id$. 
In the limit of $N\to \infty$ it follows from Eq. \eqref{grouplimit} that 
\begin{equation}
\label{idealH}
\mathcal{H}=\frac{\hbar \omega}{2} \sigma^{(1)}_z  \otimes \id^{(2)} + \frac{\hbar \omega}{2} \id^{(1)} \otimes 
\sigma^{(2)}_z,
\end{equation}
which is the original model Hamiltonian minus the interaction between the two qubits, just as we intended.

For the self-adjoint error operator $B$ we consider the most general form:
\begin{eqnarray}
\label{formofB2}
B&=&B^{(1)} \otimes B^{(2)}, \nonumber \\
B^{(1)}&=&\left(\alpha_0 \id^{(1)} +\alpha_x \sigma^{(1)}_x + \alpha_y \sigma^{(1)}_y+ \alpha_z \sigma^{(1)}_z
\right),\nonumber \\
B^{(2)}&=&\left(\beta_0 \id^{(2)} +\beta_x \sigma^{(2)}_x + \beta_y \sigma^{(2)}_y+ \beta_z \sigma^{(2)}_z\right),
\nonumber \\ 
&~&\alpha_0, \alpha_x, \alpha_y, \alpha_z, \beta_0, \beta_x, \beta_y, \beta_z \in \mathbb{R}.
\end{eqnarray}  

Let us first study the effects of the error operator $B$ in the limit of continuous control $N \to \infty$. 
Substituting $B$ into Eq. \eqref{B} we obtain
\begin{equation}
\label{formofB}
\mathcal{B}= \left(\alpha_0 \id^{(1)} + \alpha_z \sigma^{(1)}_z\right) \otimes B^{(2)}.
\end{equation}
This implies the time evolution
\begin{equation}
\label{eqcommut}
d\rho=-i\left[\mathcal{H}, \rho\right]dt-\frac{\gamma}{2}
\left[\cB, \left[\cB,\rho \right]\right]dt-i\sqrt{\gamma} \left[\cB,\rho\right] dW_t. 
\end{equation}

We assume an initially separable state $\rho(0)=\rho_1 \otimes \rho_2$, where we try to protect the free evolution of 
$\rho_1(t)$ governed by the first term of the Hamiltonian \eqref{toymodel}. In the case of ideal unitary pulses ($B=0$) the whole system dynamics 
is governed 
by the Hamiltonian of Eq. \eqref{idealH} (in the limit of $N\to \infty$), so this is indeed achievable.
With active errors, the dynamics for $\rho_1(t)=\mathrm{Tr}_2\{ \rho(t)\}$, with $\rho(t)$ a solution of Eq. \eqref{eqcommut}, will depend 
on the nature of $\cB$, and may display a complicated time evolution.

First, we note that in the case of a specific subset of initial conditions, namely
\begin{equation}
\label{def}
\rho(0)=\rho_1 \otimes \rho_2,\,\, \forall \rho_1\in 
\big\{\rho\big|[\rho,\sigma_z]=0\big\}, 
\end{equation}
there is no effect induced by the imperfections.
This is a direct consequence of Eq. \eqref{eqcommut} and the form of $\cB$ in Eq. 
\eqref{formofB}. $\rho_2$ is taken as an arbitrary density matrix of system $2$. In the following we consider a simple 
example for the initial condition 
\begin{eqnarray}
\rho(0)&=&\ket{\Psi}\bra{\Psi} \otimes \rho_2 ,\,\ket{\Psi}=(\ket{0}_1+\ket{1}_1)/\sqrt{2}, \nonumber \\
\rho_2&=&\frac{1}{2} \id^{(2)}, 
\end{eqnarray}
which does not fall into the definition of Eq. \eqref{def}. We determine the fidelity between the ideal evolution of the pure state 
\begin{equation}
\ket{\Psi(t)}=\ee^{-i \omega \sigma^{(1)}_z t/2} \Big(\ket{0}_1+\ket{1}_1\Big)/\sqrt{2}
\end{equation}
and the mixed state $\mathrm{Tr}_2\{ \rho(t)\}$ given by    
\begin{eqnarray}
F(t)=\sqrt{\bra{\Psi(t)} \mathrm{Tr}_2\{ \rho(t)\} \ket{\Psi(t)}}.
\end{eqnarray}  
\begin{figure}[t]
\begin{center}
\includegraphics[width=8.cm]{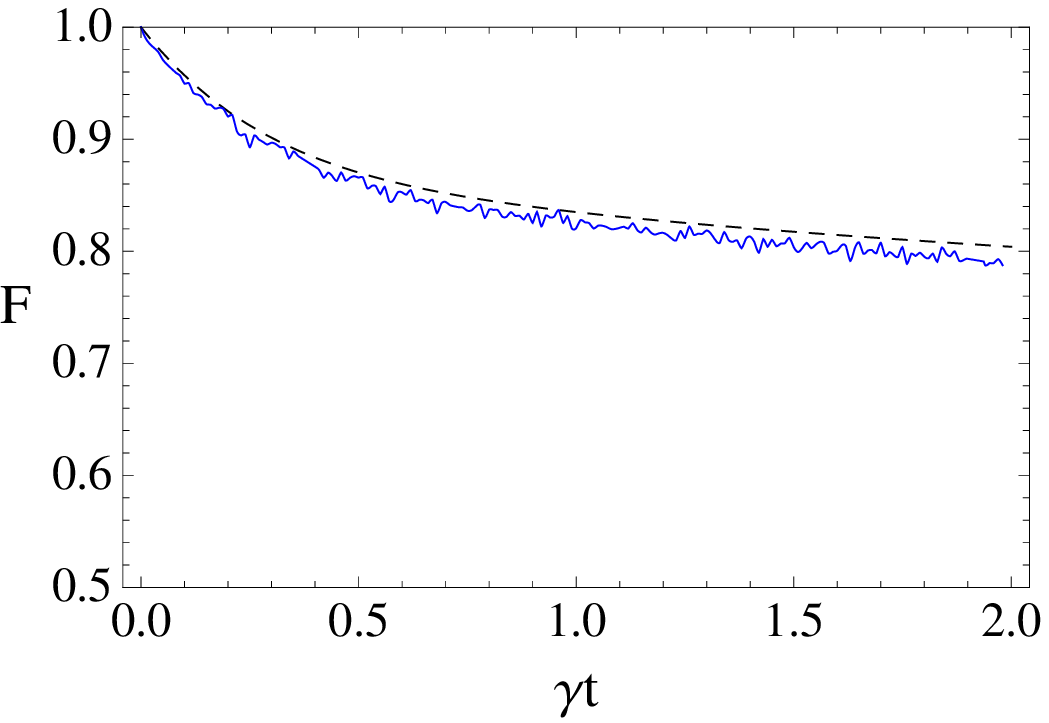}
\includegraphics[width=8.cm]{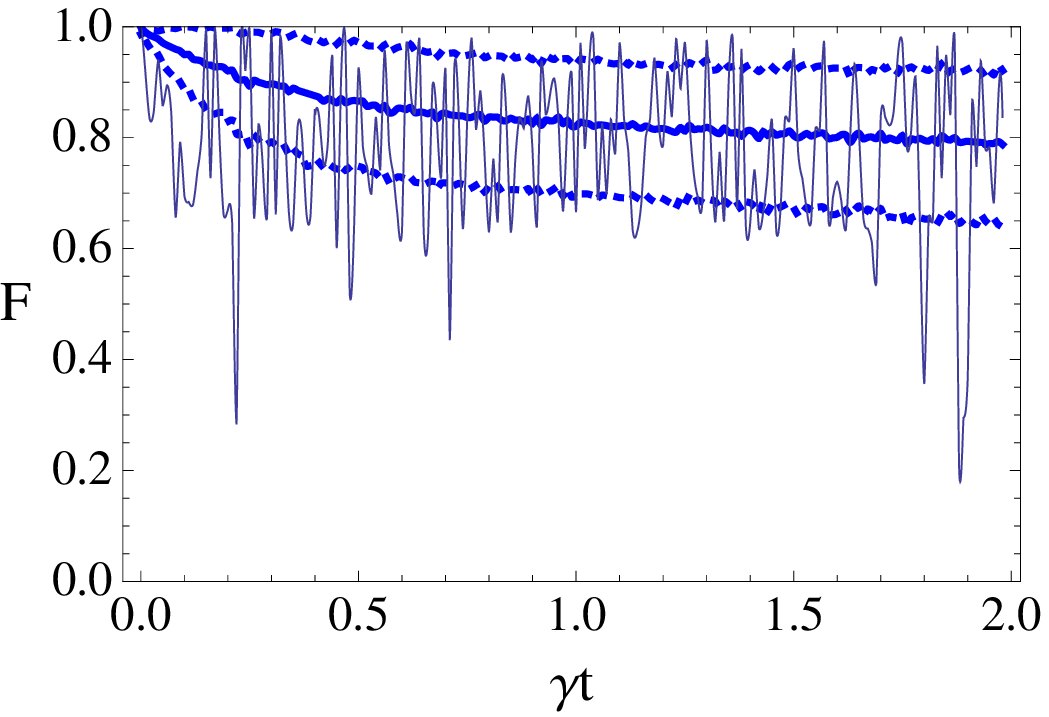}
\caption{\label{Fig2}(Color online) Dependence of fidelity $F$ on $\gamma t$. The first figure shows the average over $1000$ realizations of Eq.
\eqref{eqcommut}, where the dashed line gives the analytical result. The second figure shows a single realization of Eq.
\eqref{eqcommut}, plotted as a solid line. The thick solid line corresponds to the average over $1000$ realizations and the thick dotted 
lines define the standard
deviation of the averaged fidelity. The parameters of Eq. \eqref{formofB}, describing the nature of the error, are set to be 
$\alpha_0=\beta_0=\beta_x=\beta_y=\beta_z=1$ and $\alpha_z=1/2$.}
\end{center}
\end{figure}
 
In Fig. \ref{Fig2}, numerical results are presented for characterizing the fidelity as a distance measure of the ideal state evolution
and the state emerging from the imperfect decoupling scheme. The average fidelity describes the average over many runs, which
must be prepared similarly. From an ensemble of realizations the average fidelity is determined by 
\begin{equation}
\bar{F}(t)=\frac{1}{M}\sum^M_{i=i}\sqrt{\bra{\Psi(t)} \mathrm{Tr}_2\{ \rho_i(t)\} \ket{\Psi(t)}},
\end{equation}
where $\rho_i(t)$ is one realization of Eq. \eqref{eqcommut} and $M$ is the size of the ensemble. The corresponding standard deviation
$\sigma_t$ is given by
\begin{equation}
\sigma^2_t=\frac{1}{M}\sum^M_{i=i}\left(\sqrt{\bra{\Psi(t)} \mathrm{Tr}_2\{ \rho_i(t)\} \ket{\Psi(t)}}-\bar{F}(t)\right)^2.
\end{equation}
Instead of using algorithms for stochastic evolution, we simply generated several realizations of the Wiener 
process $W_t$ and substituted into the integrated form of Eq. \eqref{eqcommut}. We found that for $1000$ different realizations the average
fidelity coincides with the analytical solution of
\begin{equation}
 \frac{d\rho}{dt}=-i\left[\mathcal{H}, \rho\right]-\frac{\gamma}{2}
\left[\cB, \left[\cB,\rho \right]\right].
\end{equation}

The figures tell us that even in the limit $N\to \infty$ there is a decay of fidelity over time due to the error $\cB$ which is not present in the 
ideal decoupling scenario. However, as we see in Eq. \eqref{formofB}, only those parts of the error operator $B$ survive which commute with $U$. 
If we choose a different decoupling operator $U = \sigma_y^{(1)} \otimes \id_2^{(2)}$, then we would find $\cH = 0$ and 
$\cB = \left(\alpha_0 \id^{(1)} + \alpha_y \sigma_y^{(1)} \right) \otimes B^{(2)}$. Depending on the actual parameter values $\alpha_i$ of $B$ 
this offers one explanation for fidelity differences between different decoupling operators as observed in experiments.

\begin{figure}[t]
\begin{center}
\includegraphics[width=8.cm]{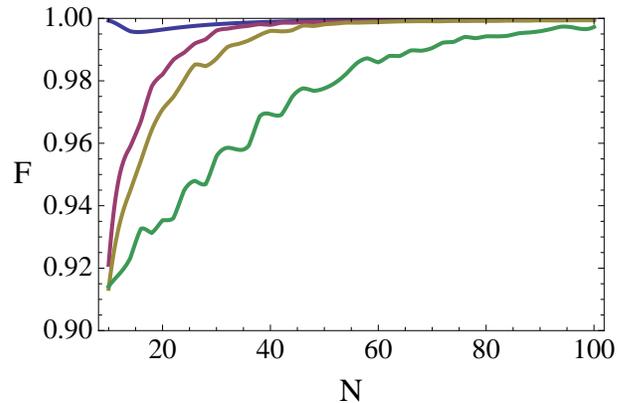}
\caption{\label{FigFinitePulses}(Color online) Dependence of fidelity on number of pulses. Plotted is the fidelity $F$ over the number of decoupling 
pulses $N$. The initial state of qubit 1 is $\rho_1=\ket{0}\bra{0}$, and the error is chosen to be $B=\sigma_y^{(1)} \otimes \sigma_y^{(2)}$. The 
remaining parameters are chosen as $\omega t = 10$ and $\gamma t \in \{0,50,100,500\}$, where the values of $\gamma t$ are in the order of 
the highest to lowest graph in the plot.}
\end{center}
\end{figure}

So far we have only studied the limit $N \to \infty$ which is obviously not achievable in experimental implementations of a dynamical decoupling 
scheme. In the following, we look at finite numbers of pulses where the noncommuting terms of the error
operator $B$ do play a role and reinduce transitions between the two qubits, with a decreasing contribution by the increase of $N$. 
We conducted a series of numerical simulations to capture the effects of the stochastic error for finite numbers of pulses. For the initial state 
of the first qubit we looked at an eigenstate of the Hamiltonian, $\ket{\Psi} = \ket{0}_1$, and at the superposition 
$\ket{\Psi} = \frac{1}{\sqrt{2}} ( \ket{0}_1 + \ket{1}_1)$. The state of qubit 2 is the totally mixed state. We investigated cases of 
commuting and noncommuting error operators, $B= \sigma_z^{(1)} \otimes \sigma_z^{(2)}$ and $B=\sigma_y^{(1)} \otimes \sigma_y^{(2)}$, respectively.
Figure \ref{FigFinitePulses} shows results for the case of the eigenstate and a noncommuting error. Plotted is the achieved fidelity 
(averaged over 1000 runs) after a fixed time $t$ depending on the number of pulses $N$ applied during that time. 
We can see that for large $N$ the fidelity approaches 1, which is in agreement with our analysis of the limit $N\to\infty$. 
For smaller $N$, however, there is a drop in the fidelity which depends on the strength of the error.

Figure \ref{FigSuperPos} shows results for the simulations which were conducted with the superposition as the initial state. The two depicted 
plots show the achieved fidelity after a fixed time $t$ depending on the number of pulses $N$ for a noncommuting and a commuting $B$, respectively. 
In the noncommuting case we approach fidelity 1 with increasing $N$ as expected, however, the drop in fidelity for lower numbers of pulses is 
stronger than in the case of the eigenstate. In the case of the commuting error we see a constant drop of the fidelity independent of $N$ which is 
expected, since the decoupling scheme has no effect on the commuting error.

The case of a commuting error $B$ acting on an initial eigenstate was not depicted because in this case the error has no effect on the qubit state.

\begin{figure}
\begin{center}
\includegraphics[width=8.cm]{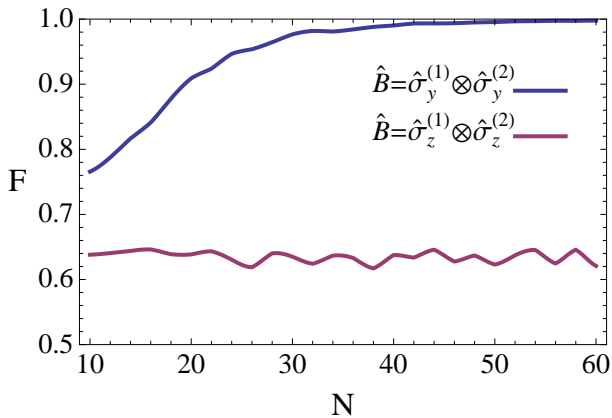}
\caption{\label{FigSuperPos}(Color online) Dependence of fidelity on number of pulses. Plotted is the fidelity $F(t)$ over the number of decoupling 
pulses $N$. The initial state of qubit 1 is $\rho_1=\frac{1}{2}(\ket{0}+\ket{1})(\bra{0}+\bra{1})$, and the error is chosen to be $B=\sigma_y^{(1)} \otimes \sigma_y^{(2)}$ (upper curve) or $B=\sigma_z^{(1)} \otimes \sigma_z^{(2)}$ (lower curve). The 
remaining parameters are chosen as $\omega/g = 10$, $\omega t = 10$ and $\gamma t = 50$.}
\end{center}
\end{figure}

The simulation results show that, with increasing $N$, we approach the dynamics of the continuous control limit which we studied earlier. As a final step in our analysis we would like to get a sense for the rate of convergence, which would allow us to predict the necessary number of pulses to reach a good approximation of the limit $N\to \infty$. Let $\rho_N(t)$ be the state emerging after time $t$ and the application of $N$ decoupling pulses.
In the following, we look at the distance between $\rho_N(t)$ and $\rho_{N+2}(t)$, the latter resulting from applying one additional cycle of our decoupling scheme. (The cycle length is $2$ as $U^2=\id$.) As a distance
measure we use the operator norm $||A||_{op}$. Our goal is to derive an upper bound for the averaged distance $\mathbb{E}\left[ ||\rho_N(t)-\rho_{N+2}(t)||_{op}\right]$ that can tell us how many pulses are needed to approach the limit $N\to\infty$. 

For the case of ideal, deterministic pulses, we found (see Appendix \ref{appendix2} for details of the derivation) 
\begin{eqnarray}
&&||\rho_N(t)-\rho_{N+2}(t)||_{op}\leqslant ||\rho(0)||_{op}\left(\ee^{4||H||_{op}t/N}-1\right)+ \nonumber \\
&&\,\,\,\,\,\,\,\,\,\,\,\,\,\,\,\,\,\,\,\,\,\,\,\,\,\,\,\,\,\,\,\,\,\,\,+2||\rho(0)||_{op}\left(\ee^{2||H||_{op}t/N}-1\right),
\end{eqnarray}
where
\begin{equation}
\rho_N(t)=\left(\prod_{k=0}^{N-1} U \ee^{-iHt/N}\right) \rho(0) \left(\prod_{k=0}^{N-1} U \ee^{-iHt/N}\right)^\dagger.
\end{equation}

In the case of stochastic decoupling pulses, the state of the system after $N$ pulses is of the form
\begin{eqnarray}
 \rho_N(t)&=&U_N(t) \rho(0) U^\dagger_N(t), \nonumber \\
U_N(t)&=& \left[ U\left(t/N^2\right) \ee^{-iHt/N}\right]^N, \nonumber \\
U(t/N^2)&=&\ee^{-i \sqrt{\gamma}B W_{t/N^2}}U.
\label{in1}
\end{eqnarray}
We were able to derive the following inequality (see Appendix \ref{appendix2} for details of the derivation)
\begin{align}
\mathbb{E}& \left[ ||\rho_N(t) -\rho_{N+2}(t)||_{op}\right] \leqslant ||\rho||_{op}\Big(\ee^{4||H||_{op}t/N} \nonumber \\
& +2\ee^{2\left[ ||H||_{op}/N + 2\gamma ||B||^2_{op}/(N+2)^2 \right] t} \nonumber \\
& +\ee^{4 \sqrt{\gamma t (N+1) }/(N+2) ||B||_{op}} - 4 \Big).
\end{align}
As expected, the upper bound goes to 0 for $N\to \infty$. Unfortunately, it turns out that the upper bound is still far above the actual average distance we retrieved from numerical simulations, particularly for the rather large values of $\omega t$ and $\gamma t$ we used in our simulations for Fig. \ref{FigFinitePulses}. This means that the approximations used in the derivation of this upper bound are too rough to capture all the relevant details of the state evolution. Although the current upper bound is of limited practical use, we present it as a first step and hope to build on it in future work to improve on the upper bound.

\subsection{A spin interacting with a bath of nuclear spins}
We now regard a system-environment model
 where a single electron spin qubit is coupled to a number $K$ of nuclear spins. 
 We consider an interaction Hamiltonian which is given by the Fermi contact hyperfine interaction \cite{semenov, coish}. Formally 
we can describe both the electron spin and the nuclear spins as simple qubits, where we label the electron spin with the index $(1)$ and 
the nuclear spins with indices from $(2)$ to
 $(K+1)$. The Hamiltonian of our model is of the form
\begin{eqnarray}
H &=&  \sum_{k=2}^{K+1} A_k \left( \sigma_x^{(1)} \sigma_x^{(k)} + \sigma_y^{(1)} 
\sigma_y^{(k)} + \sigma_z^{(1)} \sigma_z^{(k)} \right)+ \nonumber \\
&+&\sum_{k=1}^{K+1} \frac{\hbar \omega_k}{2} \sigma_z^{(k)}
\end{eqnarray}
where the hyperfine coupling constant with index $k$ is given by $A_{k}\sim\left|\psi(\mathbf{r}_{k})\right|^{2}$. $\psi(\mathbf{r}_{k})$
is the electron envelope wave-function, evaluated at $\mathbf{r}_{k}$, the position of the $k^{\mathrm{th}}$ nuclear spin. 
$\omega_k$ is the split frequency for each individual spin qubit. From our perspective we consider the electron spin qubit as the system we want to 
protect, whereas the nuclear spins should be regarded as an environment whose influence on the electron spin we want to eliminate. 
Interactions between the nuclear spins are not considered as they have no direct effect on the electron spin.

This model is a natural extension of the two-qubit case we studied in the previous section. The difference is that our system qubit is now coupled to 
more than one other qubit, and the interaction involves an additional $\sigma_z \otimes \sigma_z$ term which we did not consider before. 
Despite these differences, the results we obtained for the two-qubit case generalize in a natural way. Analogously to Eqs. \eqref{tq_pulse} and 
\eqref{formofB2}, we choose for the decoupling operator $U$ and the error operator $B$
\begin{align}
U &= \sigma_z^{(1)} \otimes \id^{(2)} \otimes \dots \otimes \id^{(K+1)}, \\
B &= B^{(1)} \otimes B^{(2)} \otimes \dots \otimes B^{(K+1)} .
\end{align} 
As we know from the two-qubit model, only those terms in $H$ and $B$ which commute with $U$ survive in the limit $N\to \infty$. Therefore, we find for 
the 
effective Hamiltonian
\begin{equation}
\cH = \sum_{k=1}^{K+1} \frac{\hbar \omega_k}{2} \sigma_z^{(k)} + \sum_{k=2}^{K+1} A_k  \sigma_z^{(1)} \sigma_z^{(k)} ,
\end{equation}
which, compared to Eq. \eqref{idealH}, includes the surviving interaction term $\sigma_z \otimes \sigma_z$. The effective error $\cB$ is still 
equivalent to Eq. \eqref{formofB} and is given by
\begin{equation}
\cB = \left(\alpha_0 \id^{(1)} + \alpha_z \sigma^{(1)}_z\right) \otimes B^{(2)} \otimes \dots \otimes B^{(K+1)}.
\end{equation} 
As before, initial states of 
the form of Eq. \eqref{def} remain unaffected over time in the limit. However, if the electron spin is in a superposition, it will be 
affected by the surviving $\sigma_z \otimes \sigma_z$ coupling to the nuclear spins, which means that the achievable state fidelity is 
limited by this interaction. We ran numerical simulations for an electron spin coupled to five nuclear spins for both cases with finite numbers of pulses. 
As we can see from the results depicted in Fig. \ref{FigNSpin1}, the eigenstate approaches fidelity 1 for larger numbers of pulses, as predicted for 
the limit $N\to \infty$. The fidelity of the superposition state, however, is severely limited by the remaining coupling.

\begin{figure}
\begin{center}
\includegraphics[width=8.cm]{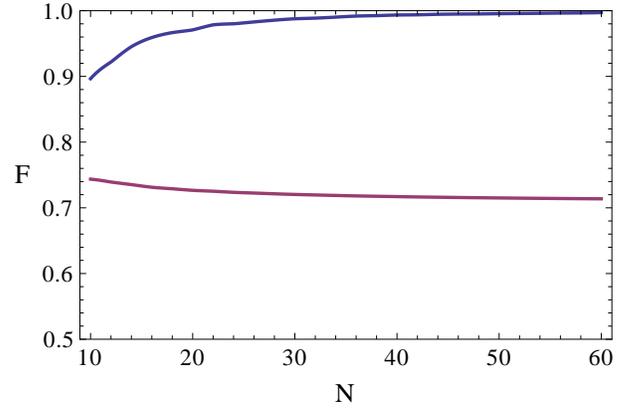}
\caption{\label{FigNSpin1}(Color online) Dependence of fidelity on number of pulses for an electron spin coupled to five nuclear spins. Plotted is the fidelity 
$F$ over the number of decoupling 
pulses $N$. The upper curve depicts behaviour of the fidelity for the electron spin in the initial state $\rho_1=\ket{0}\bra{0}$, whereas the lower 
curve starts with the electron spin in the state $\rho_1=\frac{1}{2}(\ket{0}+\ket{1})(\bra{0}+\bra{1})$. The error is chosen to be 
$B=\sigma_x^{(1)} \otimes \sigma_x^{(2)} \otimes \dots \otimes \sigma_x^{(K+1)}$ with $\gamma t = 5$. The 
remaining parameters are chosen as $\omega t = 2$ and (compare \cite{coish}) $A_k/\omega=\exp\left[ -\left(\frac{k}{5}\right)^{1/3} \right] $.}
\end{center}
\end{figure}

In order to get rid of the remaining nuclear spin interaction in $\cH$, we need a two-pulse decoupling sequence. This new sequence consists of
 alternating applications of any two Pauli pulses, e.g., $\sigma_x$ and $\sigma_y$ or $\sigma_x$ and $\sigma_z$, on the electron spin. 
We choose the sequence
\begin{equation}
\sigma^{(1)}_x \otimes \id^{\otimes K} \to \sigma^{(1)}_z \otimes \id^{\otimes K} \to \sigma^{(1)}_x \otimes \id^{\otimes K} \to \cdots ,
\end{equation}
which means that the cycle length is $4$. Looking at the limits and using Eqs. \eqref{grouplimit} and \eqref{B}($M=4$), we now find
\begin{eqnarray}
\cH =\cB=0
\end{eqnarray}
which is independent of the choice of error operator $B$. 

This can be easily understood if we consider the following simple example: Consider an arbitrary operator 
$A=a_0 \id + a_x \sigma_x +a_y \sigma_y+a_z\sigma_z$ and pulse sequences $\sigma_x$ and $\sigma_z$. The order of application is the same as above:
$\sigma_x \to \sigma_z \to \sigma_x \to \cdots$. Equation \eqref{generallimit} for $U=\sigma_x\sigma_z$ results in
\begin{equation}
 \mathcal{A}=\mathcal{P}\left(\sigma_x A \sigma_x +A\right),
\end{equation}
where $\mathcal{P}$ projects onto $\{X \in \mathcal{M}_2|\sigma_yX=X\sigma_y\}$ ($\mathcal{M}_2$ is the set of all $2\times2$ complex matrices). First,
we observe that
\begin{equation}
\sigma_x A \sigma_x +A=2 a_0 \id+ 2 a_x \sigma_x
\end{equation}
and now applying the projection map we obtain
\begin{equation}
\mathcal{P}\left(2a_0 \id+ 2a_x \sigma_x\right)=2a_0 \id.
\end{equation}
If we switch the order of pulses to $\sigma_z \to \sigma_x \to \sigma_z \to \cdots$ then we arrive at
\begin{equation}
\mathcal{P}\left(2a_0 \id+ 2a_z \sigma_z\right)=2a_0 \id
\end{equation}
which shows that the order is irrelevant.

This means that, in the limit $N \to \infty$, this decoupling scheme not only eliminates 
the influence of the nuclear spins on the electron spin completely, it is also robust against our stochastic noise model. Figure \ref{FigNSpin2} 
shows the performance of this scheme for both an eigenstate and a superposition state, depending on the number of 
pulses $N$. As we see, the new scheme is a vast improvement in the superposition case and is able to achieve a fidelity close to 1. For the eigenstate it also approaches a fidelity of 1; however, compared with the single-pulse decoupling, it is not an improvement. The lesson here is that it is beneficial to keep decoupling schemes as simple as possible for a specific application.

\begin{figure}
\begin{center}
\includegraphics[width=8.cm]{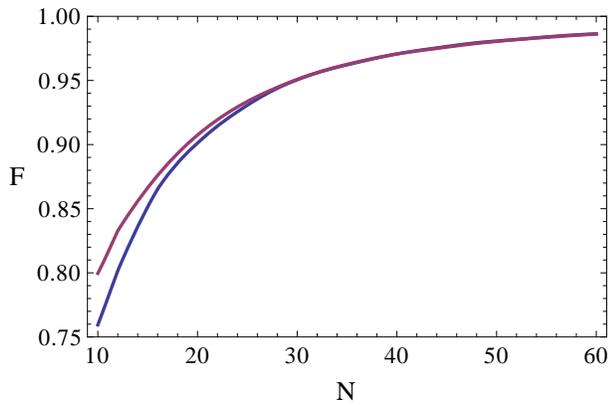}
\caption{\label{FigNSpin2}(Color online) Dependence of fidelity on number of pulses for an electron spin coupled to five nuclear spins. Plotted is the fidelity $F$ 
over the number of decoupling 
pulses $N$ for the case of the two-pulse 
decoupling scheme. The lower (blue) curve shows results for an eigenstate, whereas the upper (purple) curve corresponds to a superposition state. All parameters are chosen identically to Fig. \ref{FigNSpin1}.}
\end{center}
\end{figure}

We also derived an upper bound for the averaged distance $\mathbb{E}\left[ ||\rho_N(t)-\rho_{N+4}(t)||_{op}\right)$ for this two-pulse scheme, where $N+4$ accounts for the fact that, this time, the length 
of the cycle is 4. Using the results of Appendix \ref{appendix2},
the following averaged inequality can be derived for $t\geqslant0$:
\begin{align}
\mathbb{E}& \left[ ||\rho_N(t)-\rho_{N+4}(t)||_{op}\right] \leqslant ||\rho||_{op}\Big(\ee^{8||H||_{op}t/N}  \nonumber \\
&+4\ee^{2\left[||H||_{op}/N + 2\gamma ||B||^2_{op}/(N+4)^2 \right]t}  \nonumber \\
&+\ee^{8\sqrt{\gamma t (N+1) }/(N+4)||B||_{op}}-6 \Big).
\end{align}

\subsection{Two coupled harmonic oscillators}

So far we have discussed in detail the effects of our noise model on systems in finite-dimensional Hilbert spaces, and we were able to use the full arsenal 
of tools
derived in the Appendixes. In the following, we consider an infinite-dimensional Hilbert space. Due to the added complexity of dealing with infinite 
Hilbert spaces, we restrict ourselves to investigating the limit of large $N$ only. 

Let us consider a model of two coupled quantum harmonic oscillators, which are defined on a countable-infinite dimensional Hilbert space, i.e., a
tensor product of two symmetric Fock spaces \cite{Fockspace}. Furthermore, the domains of creation and annihilation operators
are dense in the Hilbert space of the harmonic oscillator and also equal, so none of the basis is excluded in the ergodic limiting procedure
(see Appendix \ref{appendix1}). We model the system by the Hamiltonian
\begin{equation}
 H=\hbar \omega_A a^\dagger a + \hbar \omega_B b^\dagger b+\hbar g \left(a^\dagger +a \right)\left(b^\dagger +b \right),
\end{equation}
where $\omega_A$ ($\omega_B$) is the frequency of oscillator $A$ ($B$) and $g$ is the strength of the interaction. $a$ ($a^\dagger$) and
$b$ ($b^\dagger$) are the creation (annihilation) operators of oscillators $A$ and $B$. 

To decouple system $A$ from system $B$, we can use a decoupling pulse of the form
\begin{equation}
\label{harmoscdec}
 U=\ee^{i \varphi a^\dagger a} \otimes \id_B, \; \varphi \in (0,\pi],
\end{equation}
a result taken from \cite{Bernad}. Here we restrict ourselves to the special case $\varphi=\pi$, so that the decoupling pulse is equal to $\mathcal{P}_A \otimes \id_B$, with $\mathcal{P}_A$ the parity 
operator \cite{Vitali} acting on system $A$.
At the limit $N \to \infty$, the joint system is governed by the Hamiltonian [see Eq. \eqref{grouplimit}]
\begin{eqnarray}
 &&\cH=\hbar \omega_A a^\dagger a + \hbar \omega_B b^\dagger b. 
\end{eqnarray}

Next, we define the nature of the error by the self-adjoint operator:
\begin{equation}
 B=\left(a^\dagger a+a^\dagger +a \right) \otimes \id_B.
\end{equation}
Substituting $B$ into Eq. \eqref{B}, we find
\begin{equation}
 \cB=a^\dagger a \otimes \id_B.
\end{equation}
One observes that in the limit $N \to \infty$ system $B$ evolves completely freely and system $A$ is governed by
\begin{eqnarray}
d\rho_A=&-&i\omega_A[a^\dagger a, \rho_A]dt-\frac{\gamma}{2} [a^\dagger a, [a^\dagger a, \rho_A]]dt \nonumber \\
&-&i\sqrt{\gamma}[a^\dagger a,\rho_A] dW_t.
\end{eqnarray}
Considering the expansion of the density operator into photon-number states $\ket{n}_A$,
\begin{equation}
 \rho_A=\sum^\infty_{n,m=0} \rho_{n,m}(t) \ket{n}_A \bra{m}_A,
\end{equation}
and averaging over all realizations, the time evolution results in the model of a phase-damped oscillator,
\begin{equation}
 \frac{d\rho}{dt}=-i\omega_A[a^\dagger a, \rho_A]-\frac{\gamma}{2} [a^\dagger a, [a^\dagger a, \rho_A]],
\end{equation}
which has the solution 
\begin{equation}
 \rho_{n,m}(t)=\ee^{-i \omega_A(n-m)t} \ee^{-(n-m)^2\frac{\gamma t}{2}}\rho_{n,m}(0).
\end{equation}
The coherence between photon-number states is damped, whereas diagonal terms are not affected by the error of the pulse. This shows that
prepared photon-number states can be protected in the context of this model. Despite the fact that two harmonic oscillators are 
decoupled for large numbers of pulses, initially prepared coherent states of system $A$ are dephased by the imperfectness of the pulses. 

The model of two coupled harmonic oscillators can be extended to a model where a harmonic oscillator (subsystem) interacts with a 
collection of independent harmonic oscillators (environment) \cite{Caldeira}. For the case of $N \to \infty$ the environment and the subsystem
can be decoupled completely by repeated application of the unitary pulse defined in Eq. \eqref{harmoscdec}. As already seen, 
the pulse error affects only states which are superpositions of photon-number states.

\section{Conclusions}

In the context of dynamical decoupling procedures we presented a model of imperfect pulses. An imperfect pulse is modeled by a linear 
quantum stochastic differential equation with a standard Wiener process, and this proposal allows us to impose an overall decoherence on every 
system subjected to an imperfect dynamical decoupling sequence. This technical tool allowed us to derive the time evolution for any quantum system under the influence of 
such imperfect decoupling schemes and to study the limit of continuous control. The latter can be connected to von Neumann's mean ergodic theorem which describes the 
limit as a projection operator. The projection operator projects the generator of the decoupling-free evolution on the commutant of the ideal pulse.
Our stochastic model is considered such that each imperfect pulse would yield the ideal pulse if the time difference between their application
is zero. This construction is the main idea behind relating robust decoupling sequences to cases when the projection of the generator results in an 
error-free unitary evolution. If the projection does not give an error-free evolution, then the imperfectness will accumulate in time. However,  
experimental setups work with finite numbers of pulses, and therefore we made an attempt to quantify the convergence towards the projection operator with
inequalities.

The motivation for this particular model can also be understood from the point of view
of an open quantum system, where the pulse generator apparatus may be considered as an ancillary system coupled to the joint system 
of the protected and of the irrelevant quantum system. It is assumed that a Markovian model is obtained during the elimination of the ancillary
system and therefore a mathematically equivalent dynamics of stochastic evolution is considered for each individual pulse. This implies that the 
self-adjoint error operator $B$ is directly originating from the interaction Hamiltonian between the ancilla system and the 
rest (protected plus irrelevant). 
It is worth mentioning that our model contains an important assumption, namely, that the stochastic equation describing 
the imperfection of the pulse gives rise to dynamics on a coarse-grained time axis, which cannot resolve effects on the time scale of the ancillary system's correlation time.
However, when increasing the number of pulses, the time between two consecutive pulses decreases. Therefore, at some point the time scale of the
ancillary system's correlation time is reached. This case is beyond the scope of our model. Thus the limit of an infinite number of pulses, $N \to \infty$, has to be understood strictly
as a mathematical limit because $N$ is upper bounded by the ratio of the available decoupling time divided by the correlation time of the ancillary system. 
For example, in Ref. \cite{Biercuk1} the total decoupling times are $5$ and $30$ ms and the ancillary system is an optical microwave system, whose correlation times are on the order of nanoseconds. Therefore, the critical number of pulses is around several thousand. The highest number of pulses used in this experiment was 1000, which is below that threshold.
In general, recent experiments employ a number of decoupling pulses in the range of 10 to 1000, for which our Markovian error model should be sufficiently detailed.

We carefully investigated the effects of our imperfect control pulses in a qubit-qubit model and presented analytical results in the limit of continuous control. 
We also studied the case of finite numbers of pulses and presented numerical simulation results demonstrating how different types of 
errors influence the evolution of different initial states. These results were then applied to a more complex model of an electron spin coupling to a 
nuclear spin bath. We also briefly discussed the case of infinite dimensional Hilbert spaces using the example of two coupled harmonic oscillators.

It is our hope that both this model and its application to the presented systems
can shed some light on experimental findings which are not in agreement with the theory of ideal dynamical decoupling. 

\begin{acknowledgments}
This work is supported by the BMBF project Q.com.
\end{acknowledgments}

\appendix

\section{Ergodic theorem}
\label{appendix1}

In this appendix we reconsider and make more rigorous the proof of Ref. \cite{Facchi} given for the convergence of the decoupling procedures.
 
We begin with some general and well-known statements. Let $\mathcal{H}$ be a Hilbert space and define a bounded linear operator 
${\cal T}$ on $\mathcal{H}$. The null space and range of 
${\cal I}-{\cal T}$, with ${\cal I}$ being the identity operation acting on $\mathcal{H}$, 
is given by
\begin{eqnarray}
 \mathrm{Ran}({\cal I}-{\cal T})&=&\{Y| Y=X-{\cal T}(X),\, \forall X \in \mathcal{H}\}, \nonumber \\
 \mathrm{Ker}({\cal I}-{\cal T})&=&\{Y| Y={\cal T}(Y),\, Y \in \mathcal{H}\}.
\end{eqnarray}
Now, we prove an elementary but useful property of these sets in the following proposition.
\begin{proposition}
\label{prop1}
If ${\cal U}:\mathcal{H} \to \mathcal{H}$ is unitary, then
\begin{equation}
\mathrm{Ker}({\cal I}-{\cal U})=\mathrm{Ran}({\cal I}-{\cal U})^\perp.
\end{equation}
\end{proposition}
\begin{proof}
 Let  ${\cal U}^\dagger:\mathcal{H} \to \mathcal{H}$ be the adjoint of ${\cal U}$ defined by
\begin{equation}
\langle {\cal U}(X),Y \rangle=\langle X,{\cal U}^\dagger (Y) \rangle,\,\,\forall X,Y \in \mathcal{H}, 
\end{equation}
implying $({\cal I}-{\cal U})^\dagger=I-{\cal U}^\dagger$, where $\langle.\,,.\rangle$ is the inner product on $\mathcal{H}$. 
An element $Y \in \mathrm{Ker}({\cal I}-{\cal U}^\dagger)$ fulfills
$({\cal I}-{\cal U}^\dagger)(Y)=0$, which implies that $\langle X,({\cal I}-{\cal U}^\dagger)(Y) \rangle=0$ for all 
$X \in \mathcal{H}$. The latest equation satisfies the relation
\begin{equation}
0=\langle X,({\cal I}-{\cal U}^\dagger)(Y) \rangle=\langle ({\cal I}-{\cal U})(X),Y \rangle,
\end{equation}
which is equivalent to the statement $Y \in \mathrm{Ran}({\cal I}-{\cal U})^\perp$. Thus 
$\mathrm{Ker}({\cal I}-{\cal U}^\dagger)=\mathrm{Ran}({\cal I}-{\cal U})^\perp$. 

The operator ${\cal I}-{\cal U}$ is a normal operator,
\begin{equation}
({\cal I}-{\cal U})^\dagger({\cal I}-{\cal U})=({\cal I}-{\cal U})({\cal I}-{\cal U})^\dagger, 
\end{equation}
which is equivalent to (see Theorem 12.12 in Ref. \cite{Rudin}), 
\begin{equation}
\mathrm{Ker}({\cal I}-{\cal U})=\mathrm{Ker}({\cal I}-{\cal U}^\dagger). 
\end{equation}
\end{proof}
\begin{theorem}[von Neumann's ergodic theorem]
\label{theorem1}
Let ${\cal U}$ be a unitary operator on the Hilbert space $\mathcal{H}$. Let ${\cal P}$ be the orthogonal projection
onto $\mathrm{Ker}({\cal I}-{\cal U})$. Then, for any $X \in \mathcal{H}$,
\begin{equation}
 \lim_{N \to \infty} \frac{1}{N} \sum_{k=0}^{N-1}{\cal U}^k(X)={\cal P}(X).
\end{equation}
\end{theorem}
\begin{proof}
First let $Z=X-{\cal U}(X) \in \mathrm{Ran}({\cal I}-{\cal U})$.  Using the norm $||\cdot||$,
\begin{equation}
 ||X||=\langle X,X \rangle^{1/2}_{\mathcal{H}},\,\,\forall X \in \mathcal{H},
\end{equation}
where $\langle \cdot,\cdot\rangle_{\mathcal{H}}$ is the inner product of $\mathcal{H}$, we have
\begin{equation}
 ||\frac{1}{N} \sum_{k=0}^{N-1}{\cal U}^k(Z)||=\frac{||X-{\cal U}^N(X)||}{N}\leqslant\frac{||X||+||{\cal U}^N(X)||}{N},
\end{equation}
where we used the subadditivity and the positive homogeneity properties of the norm. The statement that ${\cal U}$ is unitary is equivalent
to $||{\cal U}(X)||=||X||$ for all $X \in \mathcal{H}$. By a repeated application of this statement, we can now conclude that
\begin{equation}
\label{notclosed}
||\frac{1}{N} \sum_{k=0}^{N-1}{\cal U}^k(Z)||\leqslant\frac{2||X||}{N} \to 0 
\end{equation}
as $N \to \infty$ for any $Z \in \mathrm{Ran}({\cal I}-{\cal U})$ . 
$||X||$ is finite, because $\mathcal{H}$ is a normed space, too. By Proposition 
\ref{prop1}, $\mathrm{Ker}({\cal I}-{\cal U})=\mathrm{Ran}({\cal I}-{\cal U})^\perp$, then it follows that
\begin{equation}
\mathcal{H}=\mathrm{Ker}({\cal I}-{\cal U})\oplus \overline{\mathrm{Ran}({\cal I}-{\cal U})}. 
\end{equation}
We require that the derivation in Eq. \eqref{notclosed} holds for all elements $Z \in \overline{\mathrm{Ran}({\cal I}-{\cal U})}$. Using the 
definition of set closure, we have that $\forall \epsilon > 0$ there exists an $X \in \mathrm{Ran}({\cal I}-{\cal U})$ such that
$||Z-X||<\epsilon$. Now, choose an $N'$ such that for any $N >N'$
\begin{equation}
||\frac{1}{N} \sum_{k=0}^{N-1}{\cal U}^k(X)||<\epsilon,
\end{equation}
then for all $N >N'$
\begin{eqnarray}
&&||\frac{1}{N} \sum_{k=0}^{N-1}{\cal U}^k(Z)||\leqslant \nonumber \\
&& \leqslant ||\frac{1}{N} \sum_{k=0}^{N-1}{\cal U}^k(Z-X)||+
||\frac{1}{N} \sum_{k=0}^{N-1}{\cal U}^k(X)|| \leqslant \nonumber \\
&&\leqslant\frac{1}{N} \sum_{k=0}^{N-1}||{\cal U}^k(Z-X)||+\epsilon. 
\end{eqnarray}
${\cal U}$ is a unitary operator, which implies that $||{\cal U}(X)||=||X||$. Then
\begin{equation}
||\frac{1}{N} \sum_{k=0}^{N-1}{\cal U}^k(Z)||\leqslant\frac{1}{N} \sum_{k=0}^{N-1}||Z-X||+\epsilon=2\epsilon. 
\end{equation}
This shows that 
\begin{equation}
\lim_{N \to \infty}||\frac{1}{N} \sum_{k=0}^{N-1}{\cal U}^k(Z)||=0 
\end{equation}
and proves the statement of the theorem for the orthogonal complement of $\mathrm{Ker}({\cal I}-{\cal U})$, thereby supposing that
$Z \in \mathrm{Ker}({\cal I}-{\cal U})$. Then,
\begin{equation}
\frac{1}{N} \sum_{k=0}^{N-1}{\cal U}^k(Z)=Z={\cal P}(Z). 
\end{equation}
Thus, the strong convergence holds on the whole $\mathcal{H}$. 
\end{proof}

Next, we will apply the theorem for dynamical decoupling schemes. Let us consider the lowest order in the Taylor expansion of the Hamiltonian,
which generates the time evolution of a system subject to multiple-pulse control. We consider an arbitrary number $M$ of pulses applied 
repeatedly. Using Eq. \eqref{HN} for $\tau \to 0$, the following limiting procedure:
\begin{eqnarray}
\label{theoremappl}
&&\lim_{L \to \infty} \frac{1}{LM} \sum^{L-1}_{i=0}\sum^{M}_{j=1} U^i \left(\prod_{n=0}^{j-1} u_{j-n}\right) H 
\left(\prod_{n=0}^{j-1} u_{j-n}\right)^\dagger \left(U^i\right)^\dagger, \nonumber \\
&&u^\dagger_n=u^{-1}_n, u_n \in {\cal B}(\mathcal{H}), \forall n \in \{1,2,\dots,M\}, \nonumber \\
&&U=\prod_{n=0}^{M-1} u_{M-n},\quad H^\dagger=H  
\end{eqnarray}
needs to be evaluated.

We consider the normed space of Hilbert-Schmidt operators, defined as
\begin{eqnarray}
{\cal B}_2(\mathcal{H})&=&\{X \in {\cal B}(\mathcal{H})|\,\,\, ||X||_2<\infty \}, \nonumber \\
||X||_2&=&\sqrt{\sum_i||Xe_i||^2}, 
\end{eqnarray}
where $\mathcal{H}$ is a separable Hilbert space, $||\cdot||$ is the norm of $\mathcal{H}$, and $(e_i)$ is an orthonormal basis of 
$\mathcal{H}$. ${\cal B}(\mathcal{H})$ is the set of all bounded linear operators defined on $\mathcal{H}$. We recall from basic 
operator theory (see Theorem VI.22 in Ref. \cite{Reed}) the fact that ${\cal B}_2(\mathcal{H})$ with the Hilbert-Schmidt inner product
\begin{equation}
 \langle X,Y \rangle=\sum_i\langle Xe_i,Ye_i \rangle=\mathrm{Tr}\{X^\dagger Y\},\,\, \forall X,Y \in {\cal B}_2(\mathcal{H})
\end{equation}
is a Hilbert space. In the case of a finite-dimensional Hilbert space $\mathcal{H}$ the two sets ${\cal B}_2(\mathcal{H})$ and
${\cal B}(\mathcal{H})$ coincide; furthermore ${\cal B}(\mathcal{H})$ is the set of all linear operators. 

We define the linear operator ${\cal U}:{\cal B}_2(\mathcal{H}) \to {\cal B}_2(\mathcal{H})$ used in Theorem 
\ref{theorem1} by ${\cal U}(X)=U X U^\dagger$ [$U$ being unitary, defined in Eq.  \eqref{theoremappl}], 
which is a unitary operator due to the relation
\begin{eqnarray}
\langle {\cal U}(X),{\cal U}(Y) \rangle&=&\mathrm{Tr}\{U X^\dagger U^\dagger U Y U^\dagger\}=\mathrm{Tr}\{U X^\dagger Y U^\dagger\} 
\nonumber \\
=\mathrm{Tr}\{X^\dagger Y\}&=&\langle X,Y \rangle,\,\,\forall X,Y \in {\cal B}_2(\mathcal{H}).
\end{eqnarray}
Immediately from Theorem \ref{theorem1} we observe that the limit in Eq. \eqref{theoremappl} is equal to
\begin{equation}
 \frac{1}{M} \sum^{M}_{j=1} {\cal P} \left(\left(\prod_{n=0}^{j-1} u_{j-n}\right) H 
\left(\prod_{n=0}^{j-1} u_{j-n}\right)^\dagger \right),
\end{equation}
where ${\cal P}$ is an orthogonal projection onto the set $\{X \in {\cal B}_2(\mathcal{H})|UX=XU\}$. We remark that if $X^\dagger=X$, then
${\cal P}(X)^\dagger={\cal P}(X)$, an important property in the case of Hamiltonians.

So far, we gave a proof for the case of Hilbert-Schmidt operators and unitary operators with an arbitrary spectrum. 
Now, we extend the proof for bounded linear operators; however, we assume that $U$ has only a point spectrum,
\begin{equation}
 U=\sum_k \lambda_k P_k, \,\,|\lambda_k|=1,\,\,\sum_k P_k={\cal I}, 
\end{equation}
and $P_k$ are orthogonal projections. We consider the following limit:
\begin{equation}
\label{limitbounded}
 \lim_{N \to \infty} \frac{1}{N}\sum^{N-1}_{i=0} U^i H (U^i)^\dagger, \; H \in {\cal B}(\mathcal{H}). 
\end{equation}
Since $U$ is unitary then $\mathrm{Ran}(U)=\mathcal{H}$, which implies that
\begin{eqnarray}
&&\mathcal{H}=\overline{\bigoplus_k \mathrm{Ran}(P_k)},\\
&&\mathrm{Ran}(P_j)\cap \mathrm{Ran}(P_k)=\emptyset,\, k\neq j. \nonumber 
\end{eqnarray}
Choosing $x_k \in \mathrm{Ran}(P_k)$ Eq. \eqref{limitbounded} yields
\begin{equation}
 \lim_{N \to \infty} \frac{1}{N}\sum^{N-1}_{i=0} U^i H (U^i)^\dagger x_k = \lim_{N \to \infty} \frac{1}{N}\sum^{N-1}_{i=0} (U \lambda^*_k)^i
H x_k.
\end{equation}
We apply Theorem \ref{theorem1} for $U \lambda^*_k={\cal U}$ being unitary and for $H x_k=X \in \mathcal{H}$, the latter being true because
$H$ is bounded. Let ${\cal P}$ be the orthogonal projection on $\mathrm{Ker}({\cal I}-{\cal U})$ and we have
\begin{equation}
\lim_{N \to \infty} \frac{1}{N}\sum^{N-1}_{i=0} (U \lambda^*_k)^i H x_k={\cal P} (H x_k). 
\end{equation}
The set $\mathrm{Ker}({\cal I}-{\cal U})$ contains all those Hilbert space elements, which fulfill the equation
\begin{equation}
 (U \lambda^*_k) y=y,\,\, y \in \mathcal{H}.
\end{equation}
Allowing that $U$ can have a degenerate spectrum, we define the index set $I_k=\{i|\lambda_k=\lambda_i,\,k\neq i\}$, and then 
\begin{equation}
\mathrm{Ker}({\cal I}-U \lambda^*_k)=\Big\{y \in \mathcal{H}| P_iy=y,\,\, i \in I_k \cup \{k\}\Big\}.
\end{equation}
The collection of $\mathrm{Ran}(P_k)$ are orthogonal subspaces on $\mathcal{H}$, so the limit in Eq. \eqref{limitbounded} results in
\begin{eqnarray}
\label{limitfinal}
&&\lim_{N \to \infty} \frac{1}{N}\sum^{N-1}_{i=0} U^i H (U^i)^\dagger=\sum_k P_k H P_k+  \nonumber \\
&&+\sum_{i,j}P_i H P_j,\; i \neq j,\; \lambda_i=\lambda_j.
\end{eqnarray}
Thus, the limit is an element of the set $\{X \in {\cal B}(\mathcal{H})|UX=XU\}$.  

In the case of an unbounded operator $H$, the domain $\mathrm{Dom}(H)\subset\mathcal{H}$.
Theorem \ref{theorem1} is applied to each $x_k \in \mathrm{Ran}(P_k) \cap \mathrm{Dom}(H) \neq \emptyset$. We get the same result as
in Eq. \eqref{limitfinal}, excluding those orthogonal projections $P_i$ which satisfy 
$\mathrm{Ran}(P_i) \cap \mathrm{Dom}(H) = \emptyset$. Now, the limit is an operator $H'$ which obeys the relation
\begin{equation}
UH'x=H'Ux,\,\; \forall \, x \in \mathrm{Dom}(H). 
\end{equation}
In summary, we proved the convergence for 
${\cal B}_2(\mathcal{H})$. We were able to extend the result to ${\cal B}(\mathcal{H})$ and to some unbounded operators 
with the condition that the unitary operator, implementing the pulse, has only a point spectrum.

We leave the problem open in the case when the unitary operators have a mixed spectrum (see the characterization of these unitaries in 
Ref. \cite{Krengel}). These mathematical subtleties do not play a role in the examples we present in the main text.  

\section{Inequalities}
\label{appendix2}

Here, we derive two inequalities in order to analyze the rate of the decoupling convergence as a function of applied pulse number.

We begin with some basic notions. For a Hilbert space $\mathcal{H}$, the set of all bounded linear operators ${\cal B}(\mathcal{H})$ has a 
special property
\begin{eqnarray}
 ||A||^2_{op}=||A^\dagger A||_{op}, \; \forall \, A \in {\cal B}(\mathcal{H}), 
 \end{eqnarray}
where $||A||_{op}=\sup\limits_{||x||=1}||Ax||$ is the operator norm and $||\cdot||$ is the norm of the Hilbert space, induced by the 
inner product. A consequence is that the operator norm is a unitarily invariant norm: 
for any unitary operator $U_1$ and $U_2$ defined on $\mathcal{H}$
\begin{equation}
\label{unitarilyinv}
 ||A||_{op}=||U_1 A U_2||_{op},\; \forall \, A \in {\cal B}(\mathcal{H}).
\end{equation}
This can be briefly proved by using the norm-preserving property of any unitary operator $U$,
\begin{equation}
 ||AU||_{op}=\sup\limits_{||x||=1}||AUx||=\sup\limits_{||x'||=1}||Ax'||=||A||_{op},
\end{equation}
in the equation
\begin{eqnarray}
 &&||A||^2_{op}=||AU_2||^2_{op}=||U^\dagger_2 A ^\dagger A U_2||_{op}= \nonumber \\
&&=|| \left(U^\dagger_2 A^\dagger U^\dagger_1\right)U_1AU_2||_{op}=||U_1AU_2||^2_{op}.
\end{eqnarray}

\begin{theorem}
\label{theorem2}
Let $\mathcal{H}$ be a Hilbert space and $(u_n)_{n \in \mathbb{N}^+}$ a set of unitary operators defined on $\mathcal{H}$. 
Consider the unitary operator
\begin{equation}
U_N(t) = \prod_{n=0}^{N-1} u_{N-n} \ee^{-iHt/N},\,t \in \mathbb{R} 
\end{equation}
with self-adjoint $H \in {\cal B}(\mathcal{H})$. Then, for any  $\rho \in {\cal B}(\mathcal{H})$
\begin{eqnarray}
&&||\rho_N(t)-\rho_{N+1}(t)||_{op}\leqslant 2||\rho||_{op}\left(\ee^{2||H||_{op}|t|/N}-1\right)+ \nonumber \\
&&+||\rho_N(t)-u_{N+1}\rho_N(t)u^\dagger_{N+1}||_{op}, \nonumber \\
&&\rho_N(t)=U_N(t) \rho U^\dagger_N(t). 
\end{eqnarray}
\end{theorem}
\begin{proof}
 Using the definition of $U_N(t)$ we have
\begin{equation}
 U_{N+1}(t) = \prod_{n=0}^{N} u_{N+1-n} \ee^{-iH\frac{t}{N}}\ee^{iHt/(N^2+N)}.
\end{equation}
Since $H \in {\cal B}(\mathcal{H})$, then we know from basic operator theory that  
\begin{equation}
 \sum^{\infty}_{n=0} \frac{H^n}{n!}=\ee^{H},
\end{equation}
because the infinite series $\sum^{\infty}_{n=0} \frac{H^n}{n!}$ converge absolutely. 
We shall now express $\ee^{iHt/(N^2+N)}$ by its infinite series. Let us define
the unitary operator
\begin{equation}
 X_{j,i}=\prod_{n=j}^{i}u_n \ee^{-iHt/N};
\end{equation}
then we have
\begin{widetext}
\begin{eqnarray}
\label{expand}
 &&\rho_{N+1}(t)=u_{N+1}\ee^{-iH\frac{t}{N}}\rho_N(t)\ee^{iH\frac{t}{N}}u^\dagger_{N+1}+
\frac{it}{N(N+1)}\Big(X_{N+1,1}[H,\rho]X^\dagger_{N+1,1}+\sum^N_{i=1} X_{N+1,i+1} [H, X_{i,1}\rho X^\dagger_{i,1}]X^\dagger_{N+1,i+1}\Big)
 \nonumber \\
&&+\frac{t^2}{2!N^2(N+1)^2}\Big(-X_{N+1,1}[H,[H,\rho]]X^\dagger_{N+1,1}-
\sum^N_{i=1} X_{N+1,i+1} [H,[H, X_{i,1}\rho X^\dagger_{i,1}]]X^\dagger_{N+1,i+1}+\dots \Big)+\dots, 
\end{eqnarray}
\end{widetext}
where the binary operation $[\cdot , \cdot$ is the commutator.

For any $\rho,A \in {\cal B}(\mathcal{H})$
\begin{eqnarray}
\label{submult2}
||A^n||_{op}&\leqslant&||A||^n_{op},  \\
||A^n \rho A^m||_{op}&\leqslant&||\rho||_{op} ||A||^{n+m}_{op}, \; \forall \, n,m \in \mathbb{N}, \nonumber
\end{eqnarray}
where the submultiplicative property of the operator norm was used. Now, using the positive 
homogeneity
\begin{equation}
 ||\lambda A||_{op}=|\lambda|\,||A||_{op}, \; \forall \, \lambda \in \mathbb{C},
\end{equation}
and the subadditivity
\begin{equation}
 ||A+B||_{op}\leqslant||A||_{op}+||B||_{op}, \; \forall \, A,B \in {\cal B}(\mathcal{H})
\end{equation}
with the aid that the operator norm is also unitarily invariant [see Eq. \eqref{unitarilyinv}],
we have  
\begin{eqnarray}
&&||\rho_N(t)-\rho_{N+1}(t)||_{op}\leqslant||\rho||_{op} \sum^\infty_{n=1}\frac{(||H||_{op}\frac{2|t|}{N})^n}{n!}+ \nonumber \\
&&+||\rho_N(t)-u_{N+1}\ee^{-iHt/N}\rho_N(t)\ee^{iH\frac{t}{N}}u^\dagger_{N+1}||_{op},
\end{eqnarray}
which is a direct consequence of Eqs. \eqref{expand} and \eqref{submult2}. 

Now expressing $\ee^{iH\frac{t}{N}}$ by its infinite series, we have the following inequality:
\begin{eqnarray}
&&||\rho_N(t)-u_{N+1}\ee^{-iH\frac{t}{N}}\rho_N(t)\ee^{iH\frac{t}{N}}u^\dagger_{N+1}||_{op}\leqslant ||\rho||_{op} \times \nonumber \\
&&\times \sum^\infty_{n=1}\frac{(||H||_{op}\frac{2|t|}{N})^n}{n!}+||\rho_N(t)-u_{N+1}\rho_N(t)u^\dagger_{N+1}||_{op}. \nonumber \\
\end{eqnarray}
It follows immediately that
\begin{eqnarray}
&&||\rho_N(t)-\rho_{N+1}(t)||_{op}\leqslant 2||\rho||_{op}\left(\ee^{2||H||_{op}|t|/N}-1\right)+ \nonumber \\ 
&&+||\rho_N(t)-u_{N+1}\rho_N(t)u^\dagger_{N+1}||_{op}. \nonumber 
\end{eqnarray}
\end{proof}
The result can be used in two different ways. First, $\rho$ is a state and $U_N(t)\rho U^\dagger_N(t)$ describes the time evolution. 
A positive operator $\rho$ with the property
\begin{equation}
 \sum_i \langle e_i, \rho e_i \rangle_{\mathcal{H}}=\mathrm{Tr}(\rho)=1
\end{equation}
is called a state. $(e_i)$ is an orthonormal basis of the separable Hilbert space $\mathcal{H}$ with inner product 
$\langle \cdot , \cdot \rangle_{\mathcal{H}}$.

 Second, the inequality also holds for the case when $\rho_N(t)=U^\dagger_N(t) O U_N(t)$ and now $O=O^\dagger$ is an 
observable (self-adjoint operator). The latter describes the time evolution of an observable $O$.

\begin{corollary}
\label{corollary}
Let $\mathcal{H}$ be a Hilbert space and $(u_n)_{n \in \mathbb{N}^+}$ a set of unitary operators defined on $\mathcal{H}$, such that
for $L \in \mathbb{N}^+$
\begin{equation}
\prod_{n=0}^{L-1} u_{k+L-n}=\id, \; \forall \, k \in \mathbb{N}.  
\end{equation} Consider the unitary operator
\begin{equation}
U_N(t) = \prod_{n=0}^{N-1} u_{N-n} \ee^{-iHt/N},\; t \in \mathbb{R} 
\end{equation}
with self-adjoint $H \in {\cal B}(\mathcal{H})$. Then, for any  $\rho \in {\cal B}(\mathcal{H})$
\begin{eqnarray}
&&||\rho_N(t)-\rho_{N+L}(t)||_{op}\leqslant ||\rho||_{op}\left(\ee^{2L||H||_{op}|t|/N}-1\right)+ \nonumber \\
&&\,\,\,\,\,\,\,\,\,\,\,\,\,\,\,\,\,\,\,\,\,\,\,\,\,\,\,\,\,\,\,\,\,\,\,+L||\rho||_{op}\left(\ee^{2||H||_{op}|t|/N}-1\right), \nonumber \\
&&\rho_N(t)=U_N(t) \rho U^\dagger_N(t). 
\end{eqnarray}
\end{corollary}

\begin{proof}
Immediate from Theorem \ref{theorem2} after using
\begin{equation}
 U_{N+L}(t) = \prod_{n=0}^{N+L-1} u_{N+L-n} \ee^{-iH\frac{t}{N}}\ee^{iHLt/(N^2+NL)} 
\end{equation}
and the inequality
\begin{widetext}
\begin{eqnarray}
||\rho_N(t)&-&\left(\prod_{n=0}^{L-1}u_{N+L-n}\ee^{-iH\frac{t}{N}}\right)\rho_N(t)
\left(\prod_{n=0}^{L-1}u_{N+L-n}\ee^{-iH\frac{t}{N}}\right)^\dagger||_{op}\leqslant L||\rho||_{op} 
 \sum^\infty_{n=1}\frac{(||H||_{op}\frac{2|t|}{N})^n}{n!}+ \nonumber \\
&+&||\rho_N(t)-\left(\prod_{n=0}^{L-1}u_{N+L-n}\right)\rho_N(t)\left(\prod_{n=0}^{L-1}u_{N+L-n}\right)^\dagger||_{op}=
L||\rho||_{op}\left(\ee^{2||H||_{op}|t|/N}-1\right). \nonumber \\
\end{eqnarray}
\end{widetext}
\end{proof}

\begin{remark}
Any two norms are equivalent in the case when $\mathcal{H}$ is finite dimensional, so convergence can be shown with respect to any norm. 
Furthermore, any Schatten $p$-norm $||\cdot||_p$ (see Ref. \cite{Bhatia}) is unitarily invariant. This implies that the operator norm 
in the above inequalities can be replaced by an arbitrary Schatten norm. 
\end{remark}

As an example we demonstrate that the above results are essential in the derivation of an inequality for the case of 
random unitary pulses which are modeled by Eq. \eqref{basic}. We also show how to deal with the Wiener process $W_t$ in order to get an 
inequality for the averaged distance between density operators in a large generated ensemble. We consider the case with cycle length $2$ ($U^2=\id$),
and the state of the system after $N$ steps of decoupling has the form
\begin{eqnarray}
 \rho_N(t)&=&U_N(t) \rho U^\dagger_N(t), \nonumber \\
U_N(t)&=& \left[ U\left(t/N^2\right) \ee^{-iHt/N}\right]^N, \nonumber \\
U(t/N^2)&=&\ee^{-i \sqrt{\gamma}B W_{t/N^2}}U.
\end{eqnarray}

The unitarily invariant property of the operator norm is exploited by dissociating $\ee^{-iHt/(N+2)} $ into
$\ee^{-iHt/N} \ee^{2iHt/(N^2+2N)}$. We express the terms $\ee^{2iHt/(N^2+2N)}$ by its infinite series. Using the properties
of the operator norm and the result of \ref{corollary}, we obtain
\begin{widetext}
\begin{eqnarray}
&&||\rho_N(t)-\rho_{N+2}(t)||_{op}\leqslant ||\rho||_{op}\left(\ee^{4||H||_{op}t/N}-1\right)+
2||\rho||_{op}\left(\ee^{2||H||_{op}t/N}-1\right)+||\rho_N(t)-V_N(t)\rho V^\dagger_N(t)||_{op}, \nonumber \\
&&V_N(t)= U^2\Big(t/(N+2)^2\Big)\left[ U\Big(t/(N+2)^2\Big) \ee^{-iHt/N}\right]^N=U^2\Big(t/(N+2)^2\Big)
U'_N(t).
\end{eqnarray} 
\end{widetext}
Outside the $N$-th-power-raised braces, i.e., $U'_N(t)$, there is a squared term. 
Replacing  $\ee^{-i \sqrt{\gamma}B W_{t/(N+2)^2}}$ by its infinite series and using the property $U^2=\id$ we have
\begin{eqnarray}
&&||\rho_N(t)-V_N(t)\rho V^\dagger_N(t)||_{op}\leqslant ||\rho_N(t)-U'_N(t)\rho U'^\dagger_N(t)||_{op} \nonumber \\
&&+2||\rho||_{op}\left(\ee^{2\sqrt{\gamma}||B||_{op}W_{t/(N+2)^2}}-1\right).
\label{in2}
\end{eqnarray}
Next we analyze the difference between the unitary operators $U_N(t)$ and $U'_N(t)$. The corresponding inequality is more complicated 
than the inequalities derived before since the Wiener process is distributed in the $N$-term product and the error operator $B$ does 
not commute with either the ideal pulse $U$ or the Hamiltonian $H$. Hence it is convenient to investigate the difference between
the random processes $W_{t/N^2}$ and $W_{t/(N+2)^2}$.

We recall the definition of the Wiener process, where it is stated that the probability distribution of 
the random variable $W_{t_1}-W_{t_2}$ for $t_1>t_2\geqslant0$ is
\begin{equation}
 f_{t_1-t_2}(x)=\frac{1}{\sqrt{2 \pi (t_1-t_2)}} \ee^{-\frac{x^2}{2(t_1-t_2)}}
\end{equation}
and the moments are
\begin{equation}
 \mathbb{E}\left[(W_{t_1}-W_{t_2})^n\right] =
      \begin{cases}
        0 & n\text{ is odd}, \\
        (n-1)!!\,(t_1-t_2)^{n/2} & n\text{ is even},
      \end{cases}
\end{equation}
where $(n-1)!!=1\times3\times\cdots\times(n-1)$. In our special case $t_1=t/N^2$ and $t_2=t/(N+2)^2$, so 
\begin{equation}
 t_1-t_2=4t\frac{N+1}{N^2(N+2)^2}.
\end{equation}
Hence, we can write for the case $t/N^3\ll1$ that the probability
\begin{equation}
P\left(|W_{t/N^2}-W_{t/(N+2)^2}| \leqslant 2\sqrt{t}\frac{\sqrt{N+1}}{N(N+2)}\right)\simeq1. 
\end{equation}
Therefore we can estimate $W_{t/(N+2)^2}$ by its upper bound $W_{t/N^2} \pm 2\sqrt{t}\frac{\sqrt{N+1}}{N(N+2)}$, and expressing every
exponential $\ee^{-i\sqrt{\gamma} B 2\sqrt{t(N+1)}/(N^2+2N)}$ by its infinite series, we have
\begin{eqnarray}
&&\mathbb{E}\left[||\rho_N(t)-U'_N(t)\rho U'^\dagger_N(t)||_{op}\right]\leqslant \nonumber \\
&&\leqslant||\rho||_{op}\left(\ee^{4\frac{\sqrt{\gamma t (N+1) }}{N+2}||B||_{op}}-1\right).
\label{in3}
\end{eqnarray}

 We finally turn to the main inequality, and taking an average over the Wiener process yields for $t\geqslant0$
\begin{widetext}
 \begin{eqnarray}
\mathbb{E}\left[ ||\rho_N(t)-\rho_{N+2}(t)||_{op}\right] \leqslant ||\rho||_{op}\left(\ee^{4||H||_{op}t/N}
+2\ee^{2\left[||H||_{op}/N + 2\gamma ||B||^2_{op}/(N+2)^2 \right]t}+\ee^{4\sqrt{\gamma t (N+1) }/(N+2)||B||_{op}}-4 \right).
 \end{eqnarray}
\end{widetext}
We note that in the limit $N \to \infty$ the right-hand side of the inequality tends to $0$. From a practical point of view this means that
for large enough numbers of applied random pulses $N\gg ||B||_{op}\sqrt{\gamma t},||H||_{op}t$ the state of the system $\rho_N(t)$ remains unchanged after the application of
new cycles of pulses.


\begin{thebibliography}{99}
\bibitem{Fockspace} The symmetric Fock space 
$\Gamma_s(\mathcal{H})=~_s\bigoplus^{\infty}_{n=0} \mathcal{H}^{\otimes n}$
is generated through the symmetric tensor product of the elements of the complex Hilbert space $\mathcal{H}$.
If $\mathcal{H}=\mathbb{C}$ then $\Gamma_s(\mathcal{H})=l^2(\mathbb{N})$, and the sequence space $l^2(\mathbb{N})$ stands 
for the Hilbert space of the harmonic oscillator, and its basis states are associated with the photon-number states.
\end{thebibliography}
\end{document}